\documentclass[runningheads]{llncs}
\usepackage{amssymb}
\usepackage{graphicx}
\usepackage{amsmath}
\usepackage{epsfig}
\usepackage{url}

\urldef{\mailsa}\path|{vamsik}@engr.uconn.edu|

\begin{document}

\title{A Simplified Proof For The Application Of Freivalds' 
Technique to Verify Matrix Multiplication}

\titlerunning{Simplified proof to verify matrix multiplication}

\author{Vamsi Kundeti}
\institute{Department of Computer Science and Engineering\\
University of Connecticut\\
Storrs, CT 06269, USA\\
\mailsa\\
} \maketitle

\begin{abstract}
{\em Fingerprinting} is a well known technique, which is often used in
designing Monte Carlo algorithms for verifying identities involving matrices,
integers and polynomials. The book by Motwani and Raghavan~\cite{motwani1995}
shows how this technique can be applied to check the correctness of matrix 
multiplication -- check if $AB=C$ where $A,B$ and $C$ are three $n\times n$ matrices. 
The result is a Monte Carlo algorithm running in time $\Theta(n^2)$ with an exponentially 
decreasing error probability after each independent iteration. In this paper we 
give a simple alternate proof addressing the same problem. We also give further 
generalizations and relax various assumptions made in the proof.
\end{abstract}

\section{Introduction}
{\em Fingerprinting} or {\em Freivalds' technique} is a standard method which is often
employed in designing Monte Carlo algorithms. Let $U$ be a large universe/set of elements,
given any $x,y\in U$ our goal is to check if $x$ and $y$ are the same. Since we need 
$\Theta(\log(U))$ bits to represent any $x,y \in U$, this means checking if $x=y$ 
deterministically would need $\Omega(\log(U))$ time. The basic idea behind finger printing
is create a random mapping $r:U\rightarrow V$ such that $|V| \ll |U|$, and verify if $V(x)=V(y)$.
However it should be clear that $V(x)=V(y)$ does not necessarily mean $x=y$ -- in fact the goal is
to find a $V$ such the {\em error probability} $P[V(x)=V(y)|x=y]$ is very small. Once we prove
that our {\em error probability} is bounded by some constant, a Monte Carlo algorithm is clearly
immediate. Motwani and Raghavan~\cite{motwani1995} applied this technique to check the correctness
of matrix multiplication, we state the as follows. Given three $n\times n$ matrices $A,B$ and $C$ 
check if $AB=C$. Clearly a simple deterministic algorithm takes $\Theta(n^3)$ time. Firstly In 
this paper we give a simple alternate proof for the Theorem-$7.2$ presented in~\cite{motwani1995},
secondly we relax various constraints and give a much general proof.

\section{Our Proofs}
We first give a simple and alternative proof for Theorem-$7.2$ in ~\cite{motwani1995}. Later in Theorem~\ref{thm2}
we show that the assumption on the {\em uniformness} is not necessary. 
\begin{theorem}
\label{thm1}
Let $A$,$B$ and $C$ be three $n\times n$ matrices such that $AB \neq C$. Let $\vec{r} \in \{0,1\}^n$
is a random vector from a uniform distribution. Then $P[AB\vec{r} = C\vec{r} | AB \neq C] \leq 1/2$
\end{theorem}
\begin{proof}
Let $X$ be a $n\times n$ matrix and $\vec{x_1}, \vec{x_2}\ldots \vec{x_n}$ be the column vectors of $X$.
Then $X\vec{r} = \sum_{i=1}^{n}r_i\vec{x_1}$. This means that multiplying a vector with a matrix is linear
combination of the columns, the coefficient $r_i$ is the $i^{th}$ component of $\vec{r}$. Since $\vec{r}$
is a boolean and $r_i$ acts as an indicator variable on the selection of column $\vec{x_i}$. So if $\vec{r}$
is chosen from a uniform distribution $P[r_i=0] = P[r_i=1] = 1/2$. 

Now let $D=AB$ and $\vec{d_1},\vec{d_2}\ldots \vec{d_n}$ be the column vectors of $D$, similarly 
let $\vec{c_1},\vec{c_2}\ldots \vec{c_n}$ be the column vectors of $C$. Let
$Y =\{\vec{d_j} | \vec{d_j}\neq \vec{c_j}$, clearly $|Y| \geq 1$ since $C\neq D$. Then
$P[AB\vec{r} = C\vec{r} | AB \neq C] = \displaystyle\Pi_{\vec{d_i}\notin Y}P[r_i] = (1/2)^{n-|Y|} \leq 1/2$
since $1 \leq |Y|\leq n-1$.  Intuitively this means we select our random vector $\vec{r}$ such that
$r_i=0$ for all $d_i \in Y$, such a selection will always ensure $AB\vec{r} = C\vec{r}$ even though
$AB\neq C$.
\end{proof}

\begin{theorem}
\label{thm2}
Let $A$,$B$ and $C$ be three $n \times n$ matrices. Let $\vec{r'} = [r_1,r_2\ldots r_n]$ any vector
with each component $r_i$ is a $i.i.d$ random variable $r_i \sim f(r)$. Then 
$P[AB\vec{r} = C\vec{r} | AB\neq C] \leq f(r)$. Where $f(r)$ is an arbitrary probability density/distribution
function.
\end{theorem}
\begin{proof}
Continuing with the proof of Theorem-~\ref{thm1} , $P[AB\vec{r} = C\vec{r}|AB\neq C] = \displaystyle\Pi_{\vec{d_i}\notin Y} P[r=r_i] \leq f(r)$.
\end{proof}

\begin{corollary}
There always exists an $\Theta(n^2)$ time Monte Carlo algorithm with exponentially decreasing error probability, for the
problem to check if $AB=C$. 
\end{corollary}

\section{Conclusions}
We give a simple and alternate proof for the proof given by Motwani~\cite{motwani1995}, to verify if $AB = C$
using a Monte Carlo algorithm. We also relax uniformness assumption made by the proof.

\bibliographystyle{splncs}
\bibliography{RAND-2009.bib}

\end{document}